\documentclass[prl,10pt,twocolumn,superscriptaddress,floatfix,showpacs]{revtex4-1}

\usepackage[utf8]{inputenc}  
\usepackage[T1]{fontenc}     %Output what you want e.g., é, ł, a, ü
\usepackage[british]{babel}  %Do hyphenation according to british english
\usepackage[sc,osf]{mathpazo}
\usepackage[scaled=0.86]{berasans}  % URL font that go well wtih palatino
\usepackage[colorlinks=true, citecolor=blue, urlcolor=blue]{hyperref}  %Hyperlinks (pink, green, blue)
\usepackage{graphicx} % Package to insert exteral figures
\usepackage[babel]{microtype}  %Improves text justification
\usepackage{amsmath,amssymb,amsthm,bm,amsfonts,mathrsfs,bbm} %Usefull math packages

\usepackage{xspace}  %Useful to add space in macros
\usepackage{pgfplots}
\usepackage{xcolor,colortbl}
\usepackage{array}
\usepackage{bigstrut}
\newtheorem{theorem}{Theorem}[section]
\newtheorem{lemma}{Lemma}[section]
\begin{document}

\title{Parity Oblivious $d$-Level Random Access Codes and Class of Noncontextuality Inequalities}

%#####################################
\author{Andris Ambainis}
%\email{andris.ambainis@lu.lv }
\affiliation{Faculty of Computing, University of Latvia, Raina bulv. 19, Riga, LV-1586, Latvia.}

\author{Manik Banik}
%\email{manik11ju@gmail.com}
\affiliation{Optics \& Quantum Information Group, The Institute of Mathematical Sciences, C.I.T Campus, Tharamani, Chennai 600 113, India.}

\author{Anubhav Chaturvedi}
%\email{anubhav70747@gmail.com}
\affiliation{National Quantum Information Centre, Gdansk, Poland.}

\author{Dmitry Kravchenko}
%\email{kravchenko@gmail.com}
\affiliation{Faculty of Computing, University of Latvia, Raina bulv. 19, Riga, LV-1586, Latvia.}

\author{Ashutosh Rai}
%\email{arai.qis@gmail.com }
\affiliation{Faculty of Computing, University of Latvia, Raina bulv. 19, Riga, LV-1586, Latvia.}

%########################################

\begin{abstract}
One of the fundamental results in quantum foundations is the Kochen-Specker \emph{no-go} theorem. For the quantum theory, the \emph{no-go} theorem excludes the possibility of a class of hidden variable models where value attribution is context independent. Recently, the notion of contextuality has been generalized for different operational procedures and it has been shown that preparation contextuality of mixed quantum states can be a useful resource in an information-processing task called parity-oblivious multiplexing. Here, we introduce a new class of information processing tasks, namely $d$-level parity oblivious random access codes and obtain bounds on  the success probabilities of performing such tasks  in any preparation noncontextual theory. These bounds constitute noncontextuality inequalities for any value of $d$. For $d=3$, using a set of \emph{mutually asymmetric biased bases} we show that the corresponding noncontextual bound is violated by quantum theory. We also show quantum violation of the inequalities for some other higher values of $d$. This reveals operational usefulness of  preparation contextuality of higher level quantum systems. 
\end{abstract}

%\pacs{03.65.Ud, 02.50.Le, 03.67.Ac}

\maketitle

%########################################
	
%\section{Introduction}
\emph{Introduction}:--Kochen-Specker (KS) theorem, along with Bell's theorem, is considered as one of the most important results in the foundations of quantum mechanics (QM). Both the theorems demonstrate the impossibility of certain types of interpretation of QM in terms of hidden variables (HV). Whereas in derivation of Bell's theorem the fundamental premises are \emph{reality} and \emph{locality} \cite{Bell'64}, in the case of KS theorem the \emph{locality} premise is replaced by the premise of \emph{noncontextuality} \cite{KS'67, Mermin'93, Brunner'14}.

The notion of contextuality is generalized recently, by Spekkens, to arbitrary operational theories, and for different experimental procedures, viz., preparation procedure, measurement procedure, and transformation procedure \cite{Spekkens'05, Harrigan'07}. While the conventional definition of contextuality address to measurement contexts, and has been studied in much depth \cite{Peres, Grudka}, by considering the generalized framework for contextuality, much recently, a number of new interesting results have been discovered \cite{Banik'14,Banik'15, Harrigan'07, Leifer'13, Spekkens'09, Ravi, Mazurek, Pusey,Chailloux'16}. In particular, the study of preparation contextuality has lead to interesting connections with certain information processing task \cite{Spekkens'09}, and the concept of Bell-nonlocality \cite{Banik'15,Chailloux'16}. In the present work, we design a new family of information processing tasks and derive noncontextuality inequalities which enables an operational depiction of preparation contextuality for mixed states in higher dimensions. 
 
\emph{Preparation contextuality}:--is defined as the impossibility of representing two equivalent preparation procedures, in an operational theory, equivalently in any ontological model. More precisely, suppose two operational preparations are equivalent in the sense that the outcome probability distributions for both these preparations are identical for all measurements, i.e., the two preparations are empirically indistinguishable. Then, a hidden variable model (ontic model) which reproduces the operational statistics is preparation noncontextual, if any two equivalent preparations provide equivalent probabilistic descriptions of the system at the ontological level \cite{Spekkens'05}. It is has been shown that, mixed quantum states (preparations) are preparation contextual \cite{Spekkens'05,Banik'14}. It is also known that preparation noncontextual models underlying quantum theory must be maximally epistemic, and these in turn must be Kochen-Specker noncontextual \cite{Leifer'13}. Moreover, the concept of preparation contextuality has been used to demonstrate nonlocality of some $\psi$-epistemic models without any use of Bell’s inequality \cite{Leifer'13,Banik'14}. Along with its foundational implications, preparation contextuality has an important ramification in certain information theoretic task. It has been shown that a particular two-party information-processing task called parity-oblivious multiplexing is powered by preparation contextuality \cite{Spekkens'09}. The authors in \cite{Spekkens'09}, derived a (preparation) noncontextuality inequality which is violated by quantum theory; it has been shown that, completely mixed state of two-level quantum system violates such a bound which provides an operational demonstration of preparation contextuality. 

In this paper we ask the question whether preparation contextuality of higher level mixed quantum states can be revealed in some operational way. Interestingly, we get an affirmative answer for this question. We define a class of information tasks, namely, \emph{parity oblivious $d$-level random access codes}, and find the optimal success probabilities (bounds) for these tasks in any preparation noncontextual theory. For any $d$, such a bound constitutes a noncontextuality inequality, and a violation of the corresponding inequality by an operational theory will imply that the theory is preparation contextual. Then, for a three level quantum system, we find a mutually asymmetric biased bases (MABB) to construct a quantum protocol which exhibits the preparation contextuality of completely mixed state of three level quantum system. We give two more protocols showing quantum violation of respective inequalities for $d=4$ and $d=5$. Our findings thereby exhibit preparation contextuality of completely mixed quantum states in higher dimensions, through operational tasks. Since the family of information processing tasks that we consider here gives preparation noncontextuality inequality in any finite dimension, it opens the possibility for operational depiction of preparation contextuality of mixed sates in any higher dimension.
In the following, we first introduce the class of information processing tasks considered in our work.  

\emph{Parity oblivious $d$-level random access codes} ($d$-PORAC)---Consider the following two-party communication task. Alice receives uniformly at random some length-$2$ string $x=x_1x_2$, where $x_n$, for $n\in \{1,2\}$, takes values from a $d$-level alphabet set $\{0,1,...,d-1\}$. Bob receives, uniformly at random, an index $y\in\{1,2\}$. Bob's task is to recover the $y^{th}$ dit (i.e., $x_y$) in Alice's string. Alice can send some information about her string to help Bob, however, there is a restriction on Alice's communication to Bob which can be stated as follows; \emph{Restriction} ({\bf R}): no information about the \emph{parity} $x_1\oplus_dx_2$ of Alice's string can be transferred to Bob, where $\oplus_d$ denotes addition modulo $d$. Let us denote Bob's guess about $x_y$ by $b$, then the average success probability in this game can be expressed as $p(b=x_y)$. 

The restriction on information transfer, i.e. {\bf R}, induces a partition over the set of all strings $\{x_1x_2 :x_1,x_2\in\{0,1,...,d-1 \}\}$, into $d$ equal parts which is defined as $\mathbb{P}_l:=\{x_1x_2~|~x_1\oplus_dx_2=l\}$, where $l\in\{0,...,d-1\}$. Then, {\bf R} implies that, no information about to which partition $\mathbb{P}_l$ Alice's string $x_1x_2$ belongs can be transferred. In what follows, we derive the optimal classical average success probability for this game. First we  prove a lemma which is crucial to obtain the classical bound for $d$-PORAC task. 
\begin{lemma}\label{lemma1}
More than $1$-dit information from Alice to Bob always carries some information about the parity $x_1\oplus_dx_2$.
\end{lemma}
\begin{proof}
A classical encoding-decoding strategy could be either randomized or deterministic. Indeed, success
probability of any randomized strategy is nothing but
a convex combination of success probabilities of several
deterministic strategies. Therefore, optimal classical bound over all possible deterministic strategies is greater than or equal to the success probability in any randomized strategy; so to obtain the classical bounds for these tasks, it is sufficient to analyze only deterministic strategies.

In a deterministic strategy, for sending more than $1$-dit information, it is necessary for Alice to encode her strings into more than $d$ number of symbols. Then, let Alice has some encoding \emph{onto} map,
\begin{equation}\label{eq1}
\mathcal{E}: \{0,..,d-1\}^2~~\longrightarrow~~\{0,...,k\},
\end{equation}       
where $d\le k\le d^2-1$. Any such encoding map partitions the set of all strings (total $d^2$ in number) into $k+1$ parts $\mathbb{E}_j$, with $0\le j\le k$. On receiving the symbol $j$ from Alice, Bob gets the information that Alice's string belongs to the partition $\mathbb{E}_j$, then, Bob will not get any information about the parity of Alice's string if and only if $\mathcal{C}(\mathbb{E}_j\cap\mathbb{P}_l)=\mathcal{C}(\mathbb{E}_j\cap\mathbb{P}_{l'})$ for all $l,l'\in\{0,...,d-1\}$, where $\mathcal{C}(\cdot)$ denotes cardinality of a set. Now, whenever $k\ge d$, there exists at least one partition, say $\mathbb{E}_{j^*}$, in which the number of strings is strictly less than $d$. This further implies that, there exist at least one partition $\mathbb{P}_{l^*}$ such that $\mathcal{C}(\mathbb{E}_{j^*}\cap\mathbb{P}_{l^*})=0$. Therefore, obtaining the symbol $j^*$ from Alice, Bob will conclude that parity of the Alice's string is not $l^*$, and as result, Bob can guess some other parity (except $l^*$) with a probability greater than $\frac{1}{d}$. This proves our claim.   
\end{proof}
According to the lemma$\ref{lemma1}$, no more than $1$-dit information is allowed from Alice to Bob, which seems to put similar restriction as in the $d$-level RAC task recently studied in \cite{Tavakoli'15}. However, in our task the restriction is more stringent; whereas in d-RAC the only restriction is that no more than $1$-dit information transfer is allowed, but in our case, $1$-dit (or even less) amount of communications that can carry information of parity is not allowed. Due to the similarity to the d-RAC task, we call our task as d-PORAC. 

It is important to note here that, the class of tasks we consider here is crucially different from the parity-oblivious multiplexing (POM) task considered in \cite{Spekkens'09}. POM is a task between two parties where some $n$-bit strings is given to the one party and the task of the other party is to guess an arbitrarily chosen single bit of the string. Additionally, a restriction is imposed on allowed communication which in turn determines the possible classical protocols. The only allowed classical protocols for POM are those that encode only a single bit (chosen arbitrarily) from the $n$-bit string. In contrast to this, for our $d$-level PORAC task the considered restriction allows $1$-dit communication in more general ways. However, even with this relaxed feature we find  
the optimal classical success probabilities.
\begin{theorem}\label{thm1}
The optimal classical success probability of d-PORAC is $1/2(1+1/d)$.
\end{theorem}
\begin{proof}
From lemma\ref{lemma1} and the discussion soon after, it is clear that d-PORAC is a restricted version of d-RAC. So the optimal classical success probability of d-PORAC can not be more than that of d-RAC. Recently, it has been shown that for d-RAC of string length $2$, the optimal classical success probability is $1/2(1+1/d)$ \cite{Ambainis'15}. The remaining argument is to show that even in d-PORAC this optimal value of d-RAC is achievable. If Alice always encodes her first (second) dit and send it to Bob, then Bob can perfectly guess about the first (second) dit and he guesses the other dit randomly; this protocol gives the required optimal average success probability same as in d-RAC. Note that this protocol does not carry any information about the parity.       
\end{proof}     

Alice and Bob can try to play this game using resources from a generalized operational theory \cite{Barrett'07,Janotta'11}. However, in the following we prove a \emph{no-go} result which states that for certain class of such theories, the success probability for the d-PORAC game is no more than the optimal classical success.

\emph{Generalized operational theory}: A generalized operational theory, as discussed in \cite{Spekkens'05,Harrigan'07}, merely specifies the probabilities $p(k|M, P)$ of different outcomes $k\in \mathcal{K}_M$ that may result from a measurement procedure $M\in \mathcal{M}$ performed on a system following some preparation procedure $P\in \mathcal{P}$; where $\mathcal{M}$ and  $\mathcal{P}$ denote the sets of measurement procedures and preparation procedures respectively, and $\mathcal{K}_M$ denotes the set of measurement results for the measurement $M$. As an example, in an operational formulation of quantum theory (QT), preparation $P$ is associated with a density operator $\rho$ on some Hilbert space, and measurement $M$ is associated with a positive operator valued
measure (POVM) $\{E_k~|~ E_k \ge 0 ~\forall~ k~\mbox{and}~\sum_kE_k = \mathbf{I}\}$. The probability of obtaining outcome $k$ is given by the Born rule, $p(k|M, P) = \mbox{Tr}(\rho E_k)$.

For playing the d-PORAC game in a generalized operational theory, Alice encodes her strings $x$ in some state (preparation) $P_x$ and sends the encoded state to Bob. For decoding $y^{th}$ dit, Bob performs some $d$ outcome measurement $M_y$ and guess the dit according to the measurement results. The average success probability can be expressed as:
\begin{equation}\label{eq2}
p(b=x_y)=\frac{1}{2\times d^2}\sum_{y\in\{1,2\}}\sum_{n\in\{0,....d-1\}^2}p(b=x_y|P_x,M_y).
\end{equation}
To satisfy the parity oblivious condition, Alice's encoding must satisfy the following relations:
\begin{eqnarray}\label{eq3}
\sum_{x\in\mathbb{P}_l}p(P_x|k,M)=\sum_{x\in\mathbb{P}_{l'}}p(P_x|k,M),~~\forall~k,M,\\\mbox{and}~~\forall~l,l'\in\{0,...,d-1\}.\nonumber
\end{eqnarray} 
Interestingly, due to this restriction the success probability of d-PORAC in any preparation noncontextual theory is restricted by the optimal classical value. 

First we give a short discussion of the general framework for ontological model of an operational theory and briefly explain the notion of preparation contextuality in these theories.

\emph{Ontological model}: In an ontological model of an operational theory, the primitives of description are the \emph{real} properties of a system, called ontic state $\lambda\in\Lambda$, where $\Lambda$ being the ontic state space. A preparation procedure $P$ yields a probability distribution
$p(\lambda|P)$ over the ontic states. Measurement $M$ performed on a
system described by ontic state $\lambda$ yields outcome $k$ with probability $p(k|\lambda,M)$. The ontological model to be compatible with the operational theory must satisfy the probability reproducibility condition, i.e., $p(k|P,M)=\int_{\lambda\in\Lambda}d\lambda p(\lambda|P)p(k|\lambda,M)$. An ontological model is preparation noncontextual, if two operational preparations yielding the same statistics for all possible measurements, also yield same distribution over the ontic states, i.e., 
\begin{eqnarray}\label{eq4}
\forall M:p(k|P,M)=p(k|P',M)\Rightarrow p(\lambda|P)=p(\lambda|P).
\end{eqnarray}
The optimal success probabilities of the $d$-PORAC games in any preparation noncontextual theory is established by the follwing theorem.
\begin{theorem}\label{thm2}
In any preparation noncontextual theory the success probability of d-PORAC can not be more than the optimal classical success probability, i.e., $1/2(1+1/d)$.
\end{theorem}
\begin{proof}
The steps in the proof of this theorem resembles to the proof of a similar theorem in \cite{Spekkens'09}, here we give a suitably modified proof for our theorem. 

In an operational theory the collection of all preparations $\mathcal{P}$ is a convex set, this enables any probabilistic mixture of preparation procedures corresponding to different states of the theory to be again a valid preparation. Consider a mixed preparation $P_l$ produced by choosing uniformly at random some preparation $P_x$ corresponding to the string $x$ belonging to the partition $\mathbb{P}_l$, i.e, $P_l=\frac{1}{d}\sum_{x\in\mathbb{P}_l}P_x$. Given the preparation $P_l$, the probability of obtaining outcome $k$ for the measurement M is,
\begin{equation}\label{eq5}
p(k|P_l,M)=\frac{1}{d}\sum_{x\in\mathbb{P}_l}p(k|P_x,M).
\end{equation}
Also the preparation $P_l$ yields the distribution on the ontic state $\lambda$,
\begin{equation}\label{eq6}
p(\lambda|P_l)=\frac{1}{d}\sum_{x\in\mathbb{P}_l}p(\lambda|P_x).
\end{equation} 
In any operational theory, the parity obliviousness puts the restriction describe in Eq.(\ref{eq3}). Using Bayes theorem and the fact that Alice's strings come from an uniform distribution, we can write,
\begin{eqnarray}\label{eq7}
\sum_{x\in\mathbb{P}_l}p(k|P_x,M)=\sum_{x\in\mathbb{P}_{l'}}p(k|P_x,M),~~\forall~k,M,\\
\mbox{and}~~\forall~l,l'\in\{0,...,d-1\}.\nonumber
\end{eqnarray}
The above expression along with Eq.(\ref{eq5}) implies $p(k|P_l,M)=p(k|P_{l'},M)$ for all $l,l'\in\{0,...,d-1\}$ and for all $k,M$. In other words different preparations $P_l$ corresponding to different partitions $\mathbb{P}_l$ are operationally equivalent. If we assume that an operational theory is preparation noncontextual then according to Eq.(\ref{eq4}), we have,
\begin{equation}\label{eq8}
p(\lambda|P_l)=p(\lambda|P_{l'}),~~\forall~~l,l'\in\{0,...,d-1\},
\end{equation}
or equivalently by using Eq.(\ref{eq6}) we can say, for all $l,l'$,
\begin{equation}\label{eq9}
\sum_{x\in\mathbb{P}_l}p(\lambda|P_x)=\sum_{x\in\mathbb{P}_{l'}}p(\lambda|P_x).%,~~\forall~~l,l'\in\{0,...,d-1\}.
\end{equation}
Applying Bayes theorem in Eq.(\ref{eq9}) we have, for all $l,l'$, 
\begin{equation}\label{eq10}
\sum_{x\in\mathbb{P}_l}p(P_x|\lambda)=\sum_{x\in\mathbb{P}_{l'}}p(P_x|\lambda).%,~\forall~l,l'\in\{0,...,d-1\}.
\end{equation}
Thus we can say that, for preparation noncontextual models, 	parity obliviousness at the operational level implies similar consequence at the level of the hidden variables, i.e, parity obliviousness should be satisfied at the hidden variable level too.

Hidden state $\lambda$ provides a classical encoding of $x$. 
But as just shown, for preparation noncontextual theories $\lambda$ cannot contain information about parity. Now the proof of this theorem follows from the result obtained in lemma\ref{lemma1} and theorem\ref{thm1}.  
\end{proof}

The theorem\ref{thm2} constitutes a class of preparation noncontextuality inequalities, i.e., if in some operational theory the success probability for d-PORAC game is more than the optimal classical success then the operational theory must be preparation contextual. As an interesting example we consider the $3$-PORAC game and show preparation contextuality for completely mixed state of a three dimensional quantum system (qutrit).

\emph{Quantum protocol for $3$-PORAC}: In $3$-PORAC task the three parity partitions of Alice's strings are $\mathbb{P}_0=\{00,12,21\}$, $\mathbb{P}_1=\{01,10,22\}$, and $\mathbb{P}_2=\{02,20,11\}$. For playing this game in quantum theory Alice encodes her string $x_1x_2$ into some quantum state $\rho_{x_1x_2}$ and send the state to Bob. The parity obliviousness requirement demands that,
\begin{equation}
\rho_{00}+\rho_{12}+\rho_{21}=\rho_{01}+\rho_{10}+\rho_{22}=\rho_{02}+\rho_{20}+\rho_{11},
\end{equation} 
If Alice encodes her strings into three orthogonal sets of states $\mathcal{A}_0=\{|\psi_{00}\rangle,|\psi_{12}\rangle,|\psi_{21}\rangle\}$, $\mathcal{A}_1=\{|\psi_{01}\rangle,|\psi_{10}\rangle,|\psi_{22}\rangle\}$, and $\mathcal{A}_2=\{|\psi_{02}\rangle,|\psi_{20}\rangle,|\psi_{11}\rangle\}$ then the above requirement is always fulfilled.

Here we consider a three dimensional pure state encoding. Denoting the computational basis of qutrit as $\{|0\rangle,|1\rangle,|2\rangle\}$, any vector $|\psi\rangle\in\mathbb{C}^3$ can be represented as $|\psi\rangle=\alpha|0\rangle+\beta|1\rangle+\gamma|2\rangle$. Alice encodes her strings as follows:
\begin{eqnarray*}
|\psi_{21}\rangle&=& |0\rangle,
|\psi_{12}\rangle= |1\rangle,
|\psi_{00}\rangle= |2\rangle;\\
|\psi_{01}\rangle &=& \frac{1}{3}(2|0\rangle+|1\rangle-2|2\rangle),\\
|\psi_{10}\rangle &=&\frac{1}{3}(|0\rangle+2|1\rangle+2|2\rangle),\\
|\psi_{22}\rangle &=& \frac{1}{3}(2|0\rangle-2|1\rangle+|2\rangle);\\
|\psi_{02}\rangle &=& \frac{1}{3}(\omega^2|0\rangle+2\omega|1\rangle+2|2\rangle),\\
|\psi_{20}\rangle &=&\frac{1}{3}(2\omega^2|0\rangle+\omega|1\rangle-2|2\rangle),\\
|\psi_{11}\rangle &=& \frac{1}{3}(2\omega^2|0\rangle-2\omega|1\rangle+|2\rangle);
\end{eqnarray*}
where $\omega$ is cube root of unity. These three sets of orthonormal vectors $\mathcal{A}_0$, $\mathcal{A}_1$, and $\mathcal{A}_2$ have the following property, each vector from any of the set has similar overlap with vectors from the remaining two sets. More precisely, for example, $|\psi_{21}\rangle$ from the set $\mathcal{A}_0$ has similar overlaps (in absolute value) with vectors from set $\mathcal{A}_1$ and the set $\mathcal{A}_2$. With set $\mathcal{A}_1$ the overlaps are, $2/3$ with $|\psi_{01}\rangle$, $|\psi_{22}\rangle$, and $1/3$ with $|\psi_{10}\rangle$, and with $\mathcal{A}_2$, the overlaps are $1/3$ with $|\psi_{02}\rangle$, and $2/3$ with $|\psi_{20}\rangle$ and $|\psi_{11}\rangle$. This feature has a resemblance to a set of mutually unbiased basis (MUB) \cite{Ivanovic'81,Wootters'89}, except that in a MUB all overlaps are equal, therefore, we call the set of bases a \emph{mutually asymmetric-biased basis} (MABB).

For decoding each of the alphabet, Bob performs a three outcome quantum measurement and guess the alphabet according to the measurement result. Given the above encoding, Bob performs  measurement $\sum_{i=0}^{2}|E_i\rangle\langle E_i|=\mathbf{I}_3$ to guess the first trit $x_1$, where
\begin{eqnarray*}
|E_0\rangle&=&\frac{1}{\sqrt{7}}\left(|\psi_{00}\rangle-|\psi_{01}\rangle+|\psi_{02}\rangle\right),\\
|E_1\rangle &=& \frac{1}{\sqrt{7}}\left(|\psi_{12}\rangle+|\psi_{10}\rangle+e^{\frac{\pi\mathbf{i}}{3}}|\psi_{11}\rangle\right),\\
|E_2\rangle &=&\frac{1}{\sqrt{7}}\left(|\psi_{21}\rangle+|\psi_{22}\rangle+e^{\frac{2\pi\mathbf{i}}{3}}|\psi_{20}\rangle\right);
\end{eqnarray*}
and for the second trit $x_2$ he performs measurement $\sum_{i=0}^{2}|F_j\rangle\langle F_j|=\mathbf{I}_3$, where
\begin{eqnarray*}
|F_0\rangle &=&\frac{1}{\sqrt{7}}\left(|\psi_{00}\rangle+|\psi_{10}\rangle-|\psi_{20}\rangle\right),\\
|F_1\rangle&=& \frac{1}{\sqrt{7}}\left(|\psi_{21}\rangle+|\psi_{01}\rangle+e^{\frac{2\pi\mathbf{i}}{3}}|\psi_{11}\rangle\right),\\
|F_2\rangle&=& \frac{1}{\sqrt{7}}\left(-|\psi_{12}\rangle+|\psi_{22}\rangle+e^{\frac{\pi\mathbf{i}}{3}}|\psi_{02}\rangle\right).
\end{eqnarray*}
For this quantum protocol, it turns out that, $|\langle E_i|\psi_{ij}\rangle|^2=|\langle F_j|\psi_{ij}\rangle|^2=7/9$ for $i,j=0,1,2$. Therefore, the average success probability $P=1/18\sum_{i,j=0,1,2}(|\langle E_i|\psi_{ij}\rangle|^2+|\langle F_j|\psi_{ij}\rangle|^2)=7/9$ which is strictly greater than the corresponding classical (noncontextual) bound, i.e., $1/2(1+1/3)=2/3$.

We also find quantum protocols for $d=4$ and $d=5$, which violate the corresponding noncontextual bounds [see the Appendix]. These results are obtained numerically by optimizing over all possible pure state encodings, respectively in $\mathbb{C}^4$ and $\mathbb{C}^5$, and all possible projective measurements for decoding. For $d=4$ and $d=5$, the obtained quantum protocol give average success probabilities taking values $0.7405$ and $0.7177$ respectively, which clearly beats the respective optimal classical (as well as noncontextual) bounds of $0.625$ and $0.6$. Note that, our exact quantum protocol for the case $d=3$ is also optimal over all possible pure state encodings in $\mathbb{C}^3$ and projective measurements on them. 

We also note that the quantum advantage increases with the dimension $d$. \cite{Spekkens'09}. For $d=2$, the ratio between the best quantum and classical success probabilities  \cite{Spekkens'09} is $1.138$.
We show that, for $d=3,4,5$, the ratios are $1.167$, $1.185$, $1.196$ respectively. 

\emph{Discussion and future directions}:--In contrast with our work, the quantum protocols for the $d=2$ case in \cite{Spekkens'09} are same as the $2\mapsto1$ and $3\mapsto1$ quantum random access code (QRAC) protocols \cite{Ambainis'99,Ambainis'02}. This fails for higher $d$: the $d$-level QRAC protocols in \cite{Tavakoli'15} for string length $2$ fail to satisfy the requirement of parity obliviousness condition (as defined in our information task) for $d=3$.
As a result, our encoding-decoding scheme is quite different from the quantum RAC protocol given in \cite{Tavakoli'15}.

Our results initiate several interesting questions. The information processing tasks defined in this work lead to noncontextuality inequalities for any finite values of $d$. We show quantum violation of these inequalities for some small values of $d$. 
At this point we conjecture that, the operational task defined in this work is sufficient to reveal preparation contextuality of maximally mixed states of any finite dimensional quantum system. 
For proving this conjecture, construction of generic quantum protocols for arbitrary values of $d$ is required. On the other hand, note that in \cite{Banik'15,Chailloux'16} the authors have found the optimal quantum violation of the noncontextuality inequality given in \cite{Spekkens'09}, and in particular in \cite{Banik'15} it has been shown that the violation of the said inequality in quantum theory is less than a hypothetical generalized probability namely \emph{box world}. Finding the optimal quantum violations of our contextuality inequalities is an open problem.

\emph{Acknowledgment}: MB acknowledges his visit at the University of Latvia and AR acknowledges his visit at the Institute of Mathematical Sciences, India, where this work has been done. MB and AR would like to thank G. Kar and S. Ghosh for fruitful discussions. AA, DK, and AR acknowledge support by the European Union Seventh Framework Programme (FP7/2007-2013) under the RAQUEL (Grant Agreement No. 323970) project, QALGO (Grant Agreement No. 600700) project, and the ERC Advanced Grant MQC. AC would like to acknowledge the grant, Harmonia 4 (Grant number:  UMO-2013/08/M/ST2/00626).

\begin{widetext}
\section{Appendix}
\subsection{Quantum protocol for $4$-PORAC game}
Here parity partitions of Alice's strings are $\mathbb{P}_0=\{00,13,31,22\}$, $\mathbb{P}_1=\{01,10,23,32\}$,  $\mathbb{P}_2=\{02,20,11,33\}$, and $\mathbb{P}_3=\{03,30,12,21\}$. Let Alice encodes her string $x_1x_2$ into some quantum state $\rho_{x_1x_2}$ and send the state to Bob. The parity obliviousness requirement demands that,
\begin{eqnarray}
\rho_{00}+\rho_{13}+\rho_{31}+\rho_{22}&=&\rho_{01}+\rho_{10}+\rho_{23}
+\rho_{32}\nonumber\\
&=&\rho_{02}+\rho_{20}+\rho_{11}+\rho_{33}\nonumber\\
&=&\rho_{03}+\rho_{30}+\rho_{12}+\rho_{21}.
\end{eqnarray}
If Alice encodes her strings into four orthonormal sets of states $\mathcal{A}_0=\{|\psi_{00}\rangle,|\psi_{13}\rangle,|\psi_{31},,|\psi_{22}\rangle\}$, $\mathcal{A}_1=\{|\psi_{01}\rangle,|\psi_{10}\rangle,|\psi_{23},,|\psi_{32}\rangle\}$, $\mathcal{A}_2=\{|\psi_{02}\rangle,|\psi_{20}\rangle,|\psi_{11},|\psi_{33}\rangle\}$ and $\mathcal{A}_3=\{|\psi_{03}\rangle,|\psi_{30}\rangle,|\psi_{12},|\psi_{21}\rangle\}$ then the above requirement is always fulfilled.

Here we consider four dimensional pure state encoding. Denoting the computational basis of $\mathbb{C}^4$ as $\{|0\rangle,|1\rangle,|2\rangle,|3\rangle\}$, any vector $|\psi\rangle=\alpha|0\rangle+\beta|1\rangle+\gamma|2\rangle+\delta|3\rangle$ can be represented as $|\psi\rangle\equiv[\alpha,\beta,\gamma,\delta]$. Alice's encoding is as follows:
\begin{eqnarray}
|\psi_{00}\rangle&=&[0,0,0,1],~|\psi_{31}\rangle=[0,0,1,0],~|\psi_{13}\rangle=[0,1,0,0],~|\psi_{22}\rangle=[1,0,0,0];
\end{eqnarray}
\begin{eqnarray}
|\psi_{01}\rangle&=&[-0.1345 + 0.0225i,  -0.2539 - 0.3035i,   0.5839 + 0.0576i,   0.6933],\nonumber\\
|\psi_{10}\rangle&=&[0.1283 - 0.0404i,   0.3662 + 0.4578i,  -0.3931 - 0.0344i,   0.6947],\nonumber\\
|\psi_{32}\rangle&=&[-0.6624 + 0.2077i,  -0.2564 - 0.3007i,  -0.5853 - 0.0330i,   0.1349],\nonumber\\
|\psi_{23}\rangle&=&[-0.6843 + 0.1143i,   0.3204 + 0.4909i,   0.3862 + 0.0849i,  -0.1366];
\end{eqnarray}
\begin{eqnarray}
|\psi_{20}\rangle&=&[-0.6194 + 0.2157i,   0.2488 + 0.2796i,   0.0007 - 0.0001i,  0.6556],\nonumber\\
|\psi_{02}\rangle&=&[-0.6191 + 0.2154i,   0.0004 + 0.0002i,  -0.3737 - 0.0291i,  -0.6556],\nonumber\\
|\psi_{11}\rangle&=&[-0.3285 + 0.1796i,  -0.5105 - 0.4114i,   0.6532 - 0.0575i,  -0.0010],\nonumber\\
|\psi_{33}\rangle&=&[0.0005 + 0.0000i,   0.4360 + 0.4899i,   0.6534 + 0.0510i,  -0.3747];
\end{eqnarray}
\begin{eqnarray}
|\psi_{30}\rangle&=&[0.3702 - 0.1379i,  -0.0719 - 0.1161i,   0.6935 + 0.0054i,  -0.5868],\nonumber\\
|\psi_{03}\rangle&=&[0.3780 - 0.1122i,  -0.4163 - 0.5556i,   0.1361 - 0.0021i,   0.5865],\nonumber\\
|\psi_{12}\rangle&=&[0.5494 - 0.2050i,   0.4360 + 0.5393i,   0.1361 - 0.0159i,   0.3954],\nonumber\\
|\psi_{21}\rangle&=&[-0.5627 + 0.1673i,   0.0751 + 0.1129i,   0.6935 + 0.0271i,   0.3941].
\end{eqnarray}
Each of the four sets $\mathcal{A}_0$, $\mathcal{A}_1$, $\mathcal{A}_2$, and $\mathcal{A}_3$ forms an orthogonal basis and hence satisfy the parity obliviousness condition. For decoding the first alphabet, Bob performs a four outcome measurement $\sum_{i=0}^{3}|E_i\rangle\langle E_i|=\mathbf{I}_4$ and guess the alphabet according to the measurement result, where 
\begin{eqnarray}
|E_0\rangle&=&[-0.2490 + 0.0899i,   0.1973 + 0.2519i,  -0.3188 - 0.0254i,  -0.8516],\nonumber\\
|E_1\rangle&=&[0.3019 - 0.1035i,   0.5355 + 0.6622i,  -0.2626 - 0.0386i,   0.3202],\nonumber\\
|E_2\rangle&=&[-0.8013 + 0.2896i,   0.2154 + 0.2354i,   0.3198 + 0.0008i,   0.2646],\nonumber\\
|E_4\rangle&=&[-0.3024 + 0.1035i,  -0.1890 - 0.1869i,  -0.8509 - 0.0298i,   0.3197];
\end{eqnarray}
and for second alphabet, Bob performs measurement $\sum_{i=0}^{3}|F_i\rangle\langle F_i|=\mathbf{I}_4$ where 
\begin{eqnarray}
|F_o\rangle&=&[0.2496 - 0.0892i,  -0.2004 - 0.2484i,   0.3187 + 0.0114i,  -0.8522],\nonumber\\
|F_1\rangle&=&[-0.3082 + 0.0849i,  -0.1800 - 0.1951i,   0.8493 + 0.0679i,   0.3186],\nonumber\\
|F_2\rangle&=&[ -0.8022 + 0.2868i,  -0.2041 - 0.2454i,  -0.3195 - 0.0067i,  -0.2650],\nonumber\\
|F_3\rangle&=&[0.3076 - 0.0847i,  -0.5244 - 0.6714i,  -0.2618 - 0.0427i,   0.3195].
\end{eqnarray}
For this above encoding-decoding we find that $P=\frac{1}{32}\sum_{i,j}Tr[\rho_{i,j}E_i+F_j]=0.7405$ while the optimal classical average success probability is $1/2(1+1/4)=0.625$.
\subsection{Quantum protocol for $5$-PORAC game}
Here parity partitions are $\mathbb{P}_0=\{00,14,41,23,32\}$, $\mathbb{P}_1=\{01,10,24,42,33\}$,  $\mathbb{P}_2=\{02,20,11,34,43\}$, $\mathbb{P}_3=\{03,30,12,21,44\}$, and $\mathbb{P}_4=\{04,40,13,31,22\}$.  The parity obliviousness conditions read
\begin{eqnarray}
\rho_{00}+\rho_{14}+\rho_{41}+\rho_{23}+\rho_{32}
&=&\rho_{01}+\rho_{10}+\rho_{24}+\rho_{42}+\rho_{33}\nonumber\\
&=&\rho_{02}+\rho_{20}+\rho_{11}+\rho_{34}+\rho_{43}\nonumber\\
&=&\rho_{03}+\rho_{30}+\rho_{12}+\rho_{21}+\rho_{44}\nonumber\\
&=&\rho_{04}+\rho_{40}+\rho_{13}+\rho_{31}+\rho_{22}.
\end{eqnarray}
Similarly to previous cases, Alice encodes her strings into five orthonormal sets of states $\mathcal{A}_0=\{|\psi_{00}\rangle,|\psi_{14}\rangle,|\psi_{41},|\psi_{23}\rangle,|\psi_{32}\rangle\}$, $\mathcal{A}_1=\{|\psi_{01}\rangle,|\psi_{10}\rangle,|\psi_{24},|\psi_{42}\rangle,|\psi_{33}\rangle\}$, $\mathcal{A}_2=\{|\psi_{02}\rangle,|\psi_{20}\rangle,|\psi_{11},|\psi_{34}\rangle,|\psi_{43}\rangle\}$,  $\mathcal{A}_3=\{|\psi_{03}\rangle,|\psi_{30}\rangle,|\psi_{12},|\psi_{21},|\psi_{44}\rangle\rangle\}$ and $\mathcal{A}_4=\{|\psi_{04}\rangle,|\psi_{40}\rangle,|\psi_{13},|\psi_{31},|\psi_{22}\rangle\rangle\}$, that in computational basis of $\mathbb{C}^5$ read,
\begin{eqnarray}
|\psi_{00}\rangle=[0,0,0,0,1],~|\psi_{41}\rangle=[0,0,0,1,0],~|\psi_{14}\rangle=[0,0,1,0,0],~|\psi_{32}\rangle=[0,1,0,0,0],~|\psi_{23}\rangle=[1,0,0,0,0];
\end{eqnarray}
\begin{eqnarray}
|\psi_{10}\rangle&=&[0.23497 - 0.07340i ,  0.16411 - 0.10914i ,  0.42583 + 0.50168i , -0.19303 + 0.15607i , -0.63712],\nonumber\\
|\psi_{01}\rangle&=&[-0.18462 + 0.06593i , -0.20731 + 0.12870i , -0.16285 - 0.18055i ,  0.51463 - 0.38890i , -0.65332],\nonumber\\
|\psi_{42}\rangle&=&[-0.24112 + 0.08538i , -0.56454 + 0.30980i , -0.12652 - 0.15098i , -0.52601 + 0.37747i , -0.24885],\nonumber\\
|\psi_{24}\rangle&=&[-0.60757 + 0.21215i , -0.21817 + 0.12492i ,  0.41531 + 0.49038i ,  0.16960 - 0.12466i ,  0.25570],\nonumber\\
|\psi_{33}\rangle&=&[0.62265 - 0.18358i , -0.57837 + 0.29870i  , 0.13628 + 0.19373i  , 0.19499 - 0.14425i ,  0.19985];
\end{eqnarray}
\begin{eqnarray}
|\psi_{20}\rangle&=&[-0.37409 + 0.52632i , -0.24065 - 0.01880i ,  0.09926 - 0.17359i , -0.07590 - 0.25674i ,  0.64274],\nonumber\\
|\psi_{02}\rangle&=&[ 0.13977 - 0.20058i ,  0.64176 + 0.05606i , -0.12928 + 0.21553i ,  0.06055 + 0.19160i ,  0.64938],\nonumber\\
|\psi_{42}\rangle&=&[-0.37619 + 0.53274i ,  0.18936 + 0.01847i , -0.12115 + 0.20752i ,  0.19401 + 0.62175i , -0.23774],\nonumber\\
|\psi_{24}\rangle&=&[0.11346 - 0.15816i , -0.65018 - 0.05766i , -0.33344 + 0.55217i  , 0.07480 + 0.22713i ,  0.25061],\nonumber\\
|\psi_{33}\rangle&=&[0.14855 - 0.19490i , -0.25265 - 0.02560i ,  0.34980 - 0.54834i  , 0.17203 + 0.61397i ,  0.21416];
\end{eqnarray}
\begin{eqnarray}
|\psi_{30}\rangle&=&[ -0.00534 - 0.24987i ,  0.47253 + 0.43074i ,  0.24297 + 0.07217i , -0.20191 + 0.01868i , -0.65066],\nonumber\\
|\psi_{03}\rangle&=&[ 0.02569 + 0.64638i  ,-0.18296 - 0.16861i , -0.19043 - 0.04582i ,  0.25060 + 0.00014i , -0.64689],\nonumber\\
|\psi_{12}\rangle&=&[-0.01010 - 0.19929i , -0.49199 - 0.42851i ,  0.63332 + 0.13808i , -0.24023 + 0.00154i , -0.23796],\nonumber\\
|\psi_{21}\rangle&=&[ 0.04174 + 0.64483i ,  0.15569 + 0.12225i ,  0.24257 + 0.04356i , -0.65035 + 0.01444i ,  0.24365],\nonumber\\
|\psi_{44}\rangle&=&[0.00276 + 0.24838i ,  0.16596 + 0.19200i ,  0.61595 + 0.19261i ,  0.64285 + 0.04428i  , 0.20539];
\end{eqnarray}
\begin{eqnarray}
|\psi_{40}\rangle&=&[0.12419 + 0.15209i , -0.04125 - 0.23659i , -0.02914 - 0.24743i ,  0.20620 - 0.62431i ,  0.63986],\nonumber\\
|\psi_{04}\rangle&=&[-0.15096 - 0.18341i ,  0.03763 + 0.20009i ,  0.06565 + 0.64204i , -0.07449 + 0.23344i ,  0.65234],\nonumber\\
|\psi_{31}\rangle&=&[-0.39083 - 0.51585i ,  0.05268 + 0.25138i , -0.05547 - 0.64174i , -0.06555 + 0.18659i  , 0.24731],\nonumber\\
|\psi_{13}\rangle&=&[-0.14877 - 0.19642i ,  0.14476 + 0.63049i ,  0.01710 + 0.20906i  , 0.19415 - 0.61023i , -0.25833],\nonumber\\
|\psi_{22}\rangle&=&[-0.40081 - 0.51459i , -0.13601 - 0.63082i ,  0.03136 + 0.24802i  , 0.08812 - 0.22522i , -0.19270].
\end{eqnarray}
Bob's first decoding measurement is $\sum_{i=0}^{4}|E_i\rangle\langle E_i|=\mathbf{I}_5$ where,
\begin{eqnarray}
|E_0\rangle&=&[-0.03309 + 0.25670i , -0.25437 - 0.07003i , -0.06867 - 0.25502i ,  0.18550 - 0.19008i , -0.85036],\nonumber\\
|E_1\rangle&=&[0.25000 + 0.08268i  , 0.03584 - 0.27017i  , 0.42688 + 0.73497i , -0.24068 - 0.09143i,  -0.26020],\nonumber\\
|E_2\rangle&=&[ 0.36296 - 0.76531i ,  0.17253 - 0.20234i,  -0.26009 + 0.03026i ,  0.21365 + 0.17299i , -0.26023],\nonumber\\
|E_3\rangle&=&[-0.22262 - 0.15015i  , 0.71573 + 0.45450i ,  0.25044 - 0.09303i , -0.15832 - 0.19830i,  -0.27073],\nonumber\\
|E_4\rangle&=&[-0.22262 - 0.15015i  , 0.71573 + 0.45450i ,  0.25044 - 0.09303i , -0.15832 - 0.19830i , -0.27073];
\end{eqnarray}
and the second decoding measurement is $\sum_{i=0}^{4}|F_i\rangle\langle F_i|=\mathbf{I}_5$ where,
\begin{eqnarray}
|F_0\rangle&=&[ -0.11375 + 0.23729i , -0.21967 - 0.13607i , -0.14148 - 0.23461i ,  0.12307 - 0.24902i ,  0.84366],\nonumber\\
|F_1\rangle&=&[0.22685 + 0.13877i  , 0.09381 - 0.24781i ,  0.26314 - 0.04382i , -0.57053 + 0.62201i ,  0.27479],\nonumber\\
|F_2\rangle&=&[ 0.26939 - 0.00519i ,  0.81916 + 0.21292i , -0.24642 + 0.11038i ,  0.24692 + 0.08601i,   0.26415],\nonumber\\
|F_3\rangle&=&[ 0.11939 - 0.84029i , -0.03406 + 0.26824i ,  0.19925 - 0.16385i , -0.21871 - 0.15566i ,  0.26064],\nonumber\\
|F_4\rangle&=&[-0.25062 - 0.06551i , -0.17506 + 0.20714i ,  0.21985 + 0.81608i , 0.16329 + 0.20819i ,  0.27390].
\end{eqnarray}
For this above encoding-decoding it becomes that $P=\frac{1}{50}\sum_{i,j}Tr[\rho_{i,j}E_i+F_j]==0.71773$ while the optimal classical average success probability is $1/2(1+1/5)=0.6$.
\end{widetext}

\end{document}